\documentclass[notitlepage,superscriptaddress,showpacs,bibnotes,aps,pra,floatfix]{revtex4-1}

\usepackage{graphicx}
\usepackage{dcolumn}
\usepackage{bm}
\usepackage{amsmath}
\usepackage{amsfonts}
\usepackage{amssymb}
\usepackage{lipsum}
\usepackage[]{units}
\usepackage{hyperref}
\usepackage{color}

\newtheorem{theorem}{Theorem}[section]
\newtheorem{lemma}[theorem]{Lemma}

\newtheorem{corollary}[theorem]{Corollary}

\newenvironment{proof}[1][Proof]{\begin{trivlist}
\item[\hskip \labelsep {\bfseries #1}]}{\end{trivlist}}
\newenvironment{definition}[1][Definition]{\begin{trivlist}
\item[\hskip \labelsep {\bfseries #1}]}{\end{trivlist}}

\newenvironment{remark}[1][Remark]{\begin{trivlist}
\item[\hskip \labelsep {\bfseries #1}]}{\end{trivlist}}

\newcommand{\qed}{\nobreak \ifvmode \relax \else
      \ifdim\lastskip<1.5em \hskip-\lastskip
      \hskip1.5em plus0em minus0.5em \fi \nobreak
      \vrule height0.75em width0.5em depth0.25em\fi}

\begin{document}

\title{Exact and Efficient Simulation of Concordant Computation}

\author{Hugo Cable}
\email{Hugo.Cable@bristol.ac.uk}
\affiliation{Centre for Quantum Photonics, H. H. Wills Physics Laboratory and Department of Electrical and Electronic Engineering, University of Bristol, Merchant Venturers Building, Woodland Road, Bristol BS8 1UB, UK}

\author{Daniel E. Browne}
\email{d.browne@ucl.ac.uk}
\affiliation{Department of Physics and Astronomy, University College London, Gower Street, London WC1E 6BT, UK}

\date{\today}

\begin{abstract}
{\it Concordant computation} is a circuit-based model of quantum computation for mixed states, that assumes that all correlations within the register are discord-free (i.e. the correlations are essentially classical) at every step of the computation.  The question of whether concordant computation always admits efficient simulation by a classical computer was first considered by B. Eastin in quant-ph/1006.4402v1, where an answer in the affirmative was given for circuits consisting only of one- and two-qubit gates.  Building on this work, we develop the theory of classical simulation of concordant computation.  We present a new framework for understanding such computations, argue that a larger class of concordant computations admit efficient simulation, and provide alternative proofs for the main results of quant-ph/1006.4402v1 with an emphasis on the exactness of simulation which is crucial for this model.  We include detailed analysis of the arithmetic complexity for solving equations in the simulation, as well as extensions to larger gates and qudits.  We explore the limitations of our approach, and discuss the challenges faced in developing efficient classical simulation algorithms for all concordant computations.
\end{abstract}

\maketitle

\section{Introduction}

Understanding the hardness of simulating quantum computation on a classical computer is a central question in quantum-computing theory. Efforts to address this question are important to help identify candidates for quantum algorithms which outperform their classical counterparts. An essential aspect is to identify classes of quantum algorithms which \textit{fail} to admit any speed-up compared to their classical counterparts, since this can give us insights into the aspects of quantum mechanics that might be responsible for any quantum computational speedup. However, there are still relatively-few general results in this area. Often insight can be gained by the development of efficient simulation methods for certain families of quantum physical processes and quantum circuits.

A prominent example is the role of entanglement in unitary circuit-based quantum computation over pure states \cite{JozsaLinden03,Vidal03}.  For this model, exponential speedup with respect to classical computation  requires that certain measures of entanglement scale with problem size. This  indicates that entanglement plays an important role in pure-state circuit quantum computation.  However, entanglement-scaling on its own does not provide a sufficient condition for a computational speedup.  For example, highly-entangling circuits using only gates from the Clifford group can be efficiently simulated via the Gottesman-Knill theorem \cite{Gottesman98}. To add further nuance, there exist  models of universal quantum computation where certain entanglement measures may remain small and even tend to zero with growing computational size \cite{VanDenNest12}. These results demonstrate that the role played by entanglement in pure-state computations is a subtle one.

Much less is known, on the other hand, about quantum computation over mixed states. For pure states, absence of entanglement (separability) implies a tensor-product state, and coincides with the absence of any correlation.  However, for mixed states separability is a much weaker constraint than being uncorrelated, and the correlations in such states can exhibit both classical and non-classical correlations.  A long-standing question is whether general unitary circuits acting on separable states can be efficiently simulated classically \cite{JozsaLinden03} (i.e. assuming that the register remains separable at every stage of the computation).  Since classically-hard probability distributions can be sampled from simple quantum circuits \cite{IQP} and linear optical networks \cite{Bosonsampling}, quantum-generated states may be hard to classically simulate even in the absence of entanglement. Indeed, classical $N$-bit probability distributions require  exponentially-many parameters for their descriptions, just like entangled pure states.

A well-studied model of mixed-state computation is the DQC1 or ``one-clean-qubit'' model, which uses a (partially)-pure control qubit and a register of qubits prepared in the maximally-mixed state \cite{KnillLaflamme98}.  In this model, the normalized trace of a unitary circuit may be well approximated by the average of measurements on the control at the output.  The role of entanglement in DQC1 was studied in  Ref.~\cite{Datta05}. It was found that entanglement in the output state, quantified using multiplicative negativity across bipartite cuts, becomes a vanishingly-small fraction of the maximum possible as the size of the register increases \cite{Datta05}.

Later studies looked at the generation of discord in the output state of DQC1 \cite{Datta08,Dakic10}, looking specifically at the correlations between the control qubit and the entire register.  For ``typical'' unitaries, defined as unitaries sampled using the Haar measure, it was found that discord remains a fixed fraction of the maximum as the number of register qubits increases \cite{Datta08}.  We remark, however, that the normalized trace of Haar-random unitaries converges to zero as the size of the matrix gets large, since the corresponding eigenvalues are uniformly-distributed phases between $0$ and $2\pi$.  Hence the output of such DQC1 computations is known in this limit, and nothing can be concluded about algorithmic speedup.  Ref.~\cite{Dakic10} also provided a condition for the generation of no discord at the output of DQC1.  Nonetheless, the relationship between  entanglement and discord in cases of the DQC1 model with apparent algorithmic speedup remains little understood.

Important progress in the study of the role of correlations in mixed-state computation was made by Eastin in Ref.~\cite{Eastin10}, where \textit{concordant computation} was defined and first analysed.  A concordant state, sometimes called a ``fully-classical'' state, is defined as a state which is diagonal in the computational basis up to local-unitary transformation. A concordant computation is one which satisfies the promise that the state of the system remains concordant after each unitary gate in the computation.  Concordant states are closely related to classical probability distributions, their only non-classical attribute being the local-unitary freedom of their density-matrix eigenbasis. Thus they can be characterised via a probability distribution and local unitaries.

Monte-Carlo simulation has been  the basis for a number of methods for efficiently simulating physical processes and quantum circuits \cite{eisertsim,emersonsim,stahlke,payashan}.
Using a Monte-Carlo method, Eastin presents a general procedure \cite{Eastin10} for simulating concordant computations which samples the output statistics. He argues that it is an efficient algorithm on a classical computer when circuits are restricted to one- and two-qubit gates.
However, there are examples of concordant computation that admit efficient simulation but which do not fall within Eastins' results.

 Consider the following example of a concordant circuit for DQC1 using gates of unrestricted size: The initial state, comprising a pure control qubit and $N$ fully-mixed register qubits, is $\left\vert +\right\rangle \!\!\left\langle +\right\vert _{1}\otimes \left(1/2^{N}\right)\openone_{2\cdots N+1}$ (where $\vert\pm\rangle=\left(\vert0\rangle\pm\vert1\rangle\right)/\sqrt{2}$) and the circuit consists of a series of controlled-gates $G_{1}\cdots G_{t}$ which are Hermitian and diagonal in the computational basis, for which
$G_k\!:\left\vert 0\right\rangle _{1}\left\vert \textbf{x}\right\rangle _{2\cdots N+1}\mapsto $ $\left\vert 0\right\rangle _{1}\left\vert {\textbf{x}}\right\rangle _{2\cdots N+1}$
and
$G_k\!:\left\vert 1\right\rangle _{1}\left\vert \textbf{x}\right\rangle _{2\cdots N+1}\mapsto $ $\left( -1\right) ^{f_{k}\left( \textbf{x}\right) }\left\vert 1\right\rangle _{1}\left\vert \textbf{x}\right\rangle _{2\cdots N+1}$
(where $f_{k}(\textbf{x})$ takes values $0$ or $1$).  The output state is then
$(1/2^N)\sum_{\textbf{x}\vert f\left( \textbf{x}\right) =0}\left\vert +\right\rangle \!\! \left\langle +\right\vert _{1}\otimes \left\vert \textbf{x}\right\rangle \!\! \left\langle \textbf{x}\right\vert _{2\cdots N+1}+(1/2^N)\sum_{\textbf{x}\vert f\left( \textbf{x}\right) =1}\left\vert -\right\rangle \!\! \left\langle -\right\vert _{1}\otimes \left\vert \textbf{x}\right\rangle \!\! \left\langle \textbf{x}\right\vert _{2\cdots N+1}$
where
$f(\textbf{x})=\sum_{k=1}^{t}f_k\left( \textbf{x}\right) $
(mod 2), and the expectation value for measurements on qubit 1 in the $\vert\pm\rangle$ basis is the average value for $(-1)^{f(\textbf{x})}$ over all bit strings $\textbf{x}$.  This computation admits a straightforward Monte-Carlo simulation by sampling measurement outcomes for pure-state trajectories given input states $\left\vert +\right\rangle _{1}\left\vert \textbf{x}\right\rangle _{2\cdots N+1}$ (where \textbf{x} is a bitstring chosen uniformly at random).

In this paper, we develop new technical tools to understand concordant computation, and use these to extend and refine Eastin's results. In particular, we prove explicitly that our simulation is exact, an essential aspect of this model, since the concordance  of a state is not preserved under arbitrarily-small perturbations \cite{Ferraro10}.

We proceed as follows:  Sec.~\ref{Sec:IntuitiveNotions} provides an informal introduction to concordant computation and some key ideas for simulating it.  This includes general features of the states involved and specific requirements for simulating computations. Sec.~\ref{Sec:StateStructureSimulation} revisits the central results of Ref.~\cite{Eastin10} using a new formalism.  It provides self-contained proofs within our revised framework.  Sec.~\ref{Sec:GateStructureAndLBF} presents a new simulation procedure that bypasses a bottleneck when identifying symmetries of the system state, which is a critical part of the procedure of Ref.~\cite{Eastin10}.  In Sec.~\ref{Sec:IncompatibleLBSolns} we explain the limitations of our simulation procedure, before concluding in Sec.~\ref{Sec:Conclusions} with some discussion of the prospects of efficient simulations for all concordant computation.

\section{Overview}
\label{Sec:IntuitiveNotions}
In this section, we provide an overview of the techniques that are developed in the rest of this paper. We start with the following formal definition for a state to be concordant:
\begin{definition}
A state $\rho$, with $N$ qudit subsystems labeled by $j$, of arbitrary dimension $d_j$, is called concordant if every qudit possesses a complete set of orthogonal rank-1 projectors $\pi^{(j)}_{k_j}$ such that
\begin{equation}
\label{eq:defconcordantstate}
\rho=\sum_{k_1,k_2,\cdots,k_N} p\!\left({k_1,k_2,\cdots,k_N}\right) \pi^{(1)}_{k_1}\otimes\pi^{(2)}_{k_2}\cdots\otimes\pi^{(N)}_{k_N}
\end{equation}
where $p\!\left({k_1,k_2,\cdots,k_N}\right)$ is a probability distribution.
\end{definition}
Concordant states have the interpretation of being the only quantum states having zero non-classical correlation with respect to any bipartition of the subsystems (see for example Refs.~\cite{CCstates}).  (Non-classical correlations can be quantified using quantum discord \cite{Discord} or a variety of related measures Ref.~\cite{Modi12}.)  The basic premise of the model called concordant computation is that an algorithm is supplied, consisting of a choice of a concordant initial state, unitary circuit, and measurements, such that the quantum state remains concordant after each gate acts.  The generation (or ``encoding'') of algorithms is not of concern here --- only the simulation (``decoding'') of algorithms which satisfy the promise of concordant states between gates.

The most basic example of concordant computation is given by probabilistic classical computation using reversible gates --- which amounts to concordant computation in the computational basis.  The expectation values for observables at the output can be evaluated efficiently by a Monte-Carlo method, which uses simulation trajectories on bit (dit) strings which are computed directly from the circuit given for the computation.  More generally, concordant computations can involve entangling gates and changes to the local basis as the computation proceeds.  To illustrate, if a CNOT gate acts on a concordant state $\rho$, for which both qubits on the support of the gate are in the
$\vert\pm\rangle$ basis, then these basis elements are permuted and $\rho$ remains concordant.  However, if the control qubit is in the $\vert\pm\rangle$ basis and the target qubit is in the computational basis, then the gate maps basis elements to Bell states, and $\rho$ may or may not remain concordant depending on its symmetries: Only if $\rho$ is invariant under
$\vert +0\rangle \leftrightarrow \vert -0 \rangle$ and $\vert +1\rangle \leftrightarrow \vert -1 \rangle$ is concordance preserved.

More generally, the symmetries of the quantum state play a central role in concordant computation.  We say that a unitary $S$ with support on qudits $b$ is a symmetry of state $\rho$ if $S \rho S^\dag=\rho$.  The collection of all such unitaries (on $b$) defines a symmetry (sub)group.  For example, when $\rho$ has fully non-degenerate eigenvalues, and is diagonal in the computational basis, the symmetries include the identity operators and any phase gate.  To make the consequences of symmetry manifest, let us write concordant states in Eq.~\ref{eq:defconcordantstate} in a different form. This form is related to the original definition of a concordant state as a state of zero discord -- see Sec.~\ref{Sec:SubsystemEigenprojectorDecomposition} for more details. Given any concordant state $\rho$, and any  partition of  the qubits (or qudits) into subsets $a$ and $b$, $\rho$ can always be written
\begin{equation}
\label{eq:basicdecomposition}
\rho=\sum_k \tilde{\rho}_k^{(a)} \otimes \Pi_k^{(b)},
\end{equation}
where $\Pi_k^{(b)}$ is a set of orthogonal projectors which are related to the computational basis by local-unitary operations, and the $\tilde{\rho}_k^{(a)}$ are (unnormalised) density operators.  After collecting terms in the sum for which $\tilde{\rho}_k^{(a)}=\tilde{\rho}_{k^\prime}^{(a)}$, this decomposition of $\rho$ becomes unique (as proved in Sec.~\ref{Sec:SubsystemEigenprojectorDecomposition}) and we call the $\Pi^{(b)}_k$  in this case {\it Full-Rank Subsystem Eigenprojectors} (FRASEs).  Any symmetry $S$ of $\rho$ with support on $b$ then satisfies $S\Pi_k^{(b)}S^\dag=\Pi_k^{(b)}$ for the corresponding FRASEs.

When the spectrum of $\rho$ is fully non-degenerate, any set of gates which implements a concordant computation must map product states to product states at every step, and efficient Monte-Carlo trajectory simulation is (trivially) possible.  However, when concordant states have degenerate eigenvalues, the problem of classically simulating concordant computation becomes non-trivial and more interesting.  The degeneracy allows the gates specified in the problem to generate  trajectories which create entanglement but leave the state concordant. A computational basis representation of such a trajectory will require exponentially-growing resources.

The degeneracy in the quantum state, and the symmetries which follow from it, therefore disrupt na\"{i}ve trajectory simulation.  Fortunately, the degeneracy itself gives rise to a new way to construct trajectories by providing for families of equivalent unitary gates that lead to the same output state.  We will say that gates $G$ and $\tilde{G}$, with support on qudits $b$, are equivalent with respect to concordant state $\rho$, if $G\rho G^\dag = \tilde{G}\rho \tilde{G}^\dag$. It is easily verified that this last equation is also equivalent to the existence of a symmetry $S$ of $\rho$ on $b$ satisfying $\tilde{G}=GS$.   Hence, the challenge for efficient classical simulation of a concordant computation is  to find  circuits of equivalent gates which define trajectories with an efficient simulation (where the computational requirements of \emph{all} steps in the procedure are accounted for).

A key insight of Eastin in  Ref.~\cite{Eastin10} is that if a gate $G$ acting  on subset of qudits $b$ maps a concordant state to a concordant state, then it is  equivalent to the following sequence of gate operations which also act only on $b$: a local unitary, a classical reversible gate, and a second local unitary.  (By classical-reversible gate we mean a gate which permutes logical basis states, e.g. a CNOT, NOT or TOFFOLI gate.) The action of the two local unitaries is to first rotate the local basis of the state into the computational basis, and then rotate it  to the new local basis for the state.  If this set of alternative gates is known then a Monte-Carlo simulation will proceed via product states, and the classical simulation will be efficient. For any concordant computation this set of gates always exists, and the challenge is to efficiently compute it.

One approach to finding this equivalent gate set would be to identify all symmetries of the state (on $b$) and then search over gates to identify those with the needed properties.  However, the symmetry identification can involve exponentially-big matrices, since it involves an exhaustive search over the full state $\rho$, leading again to inefficient simulation.  Furthermore, any attempt to find trajectories that relies on testing on the whole quantum state (e.g. direct application of the criteria for classicality given in Ref.~\cite{Chen11}) will typically suffer from exponentially-scaling overheads.  As noted by Eastin, one can show that the identification of symmetries of an initially-uncorrelated state after a circuit of unitary gates (even a set of classical reversible gates) is NP hard in general.  Furthermore it is widely believed that not even universal quantum computers can solve NP-complete problems in polynomial time.

Nonetheless, Eastin argued that all necessary symmetry identification can be performed efficiently for concordant computations comprising circuits of one and two-qubit gates, together with some restriction on the form of the initial state (see Sec.~\ref{Sec:MonteCarloForConcordantComputation} for an alternative proof of this).  Note that, although a universal gate set can be obtained using only one-qubit and two-qubit gates \cite{NielsenChuang}, it does not follow that concordant computations comprising three-and-higher qubit gates always (efficiently) decompose into concordant computations with one-qubit and two-qubit gates. Furthermore, the question of when qudit-based concordant computation for $d>2$ (qudits) admits efficient classical simulation has remained entirely open so far.

In this paper, we develop a new approach to simulating concordant computation which allows us, in many cases, to go beyond the limitations discussed above. This will be described in detail in Sec.~\ref{Sec:GateStructureAndLBF} and Sec.~\ref{Sec:IncompatibleLBSolns}.  The key new idea is that it is often not necessary  to acquire full knowledge of the symmetries of the quantum state in order to identify classically-efficiently-simulable trajectories. In practice, the strict requirements upon the gate to leave the state concordant will often allow suitable trajectories to be extracted from individual quantum gates alone, side-stepping the NP-hard bottleneck in Eastin's algorithm with a tractable analysis on individual unitary gates.

The heart of our algorithm is a sub-routine that we call the \textit{Local-Basis Finder} (LBF).  In Sec.~\ref{Sec:ExplanationOfLBF} we provide a detailed analysis of the LBF, which aims to identify the local basis rotation which forms the first of the three unitary gates (local rotation, classical reversible gate, local rotation) which act equivalently to the unitary gate applied in the circuit. Knowing this local rotation, the other two gates needed for the efficient trajectory simulation can be efficiently derived.

We emphasise that a critical issue for the simulation of concordant computation is the effect of  numerical errors, such as rounding errors.  The set of concordant states has zero volume relative to the Hilbert space of all quantum states.  Small perturbations on concordant states will generate non-classical correlations (discord) \cite{Ferraro10}, and necessarily disrupt the state symmetries which Eastin's algorithm computes at every step. Thus simulation algorithms of this type have no tolerance to such errors, and they must therefore proceed via exact numerical calculation. We remark that even  the new algorithms introduced in this paper, where computing state symmetries is not always necessarily, require an exact representation of the local basis of the state for their successful implementation.

While Eastin did not consider this issue in \cite{Eastin10}, we show that his approach can be adapted to exact arithmetic while remaining efficient.   The adoption of exact arithmetic is a non-trivial and a key technical contribution of our work. Exact simulation means that one cannot use the approximate floating-point arithmetic typically used to approximate real or complex numbers in Physics simulations. We achieve this by adopting a combination of exact integer arithmetic and exact arithmetic on algebraic numbers. While the latter is not typically efficient \cite{CheeYap,LiLutzer04}, we show in Sec.~\ref{Sec:ExactnessAndEfficiency} that this computational cost is a fixed overhead and does not affect the scaling of our algorithm (and our exact version of Eastin's algorithm).

\section{Structure of concordant states and equivalent circuits for simulating concordant computation}
\label{Sec:StateStructureSimulation}

In this section, we revisit the key results of Ref.~\cite{Eastin10}, using an alternative argument with some new techniques that clarify how the approach works.  In Sec.~{\ref{Sec:SubsystemEigenprojectorDecomposition} we introduce new tools for understanding a notion of degeneracy which plays a central role in Ref.~\cite{Eastin10} and in the current work.  Then, in Sec.~\ref{Sec:MonteCarloForConcordantComputation}, we formally derive the general method for trajectory-based simulation of concordant computation using these tools.  We also review the difficulties encountered in Ref.~\cite{Eastin10} which centre around a step in the simulation algorithm --- termed ``Diagnosing the degeneracy'' --- which attempts to identity symmetries of the system state.

\subsection{Subsystem-eigenprojector decomposition for quantum-classical and concordant states}
\label{Sec:SubsystemEigenprojectorDecomposition}

We begin by introducing a key notion of classicality relevant for concordant computation.  A state $\rho$, with subsystems labeled $a$ and $b$, is said to be {\it classical with respect to $b$} if there exists a complete set of rank-one projectors $\{\pi^{(b)}_k\}$ on $b$ such that $\rho=\sum_k \pi^{(b)}_k \rho \pi^{(b)}_k$, or equivalently
$\rho=\sum_k p_k \rho^{(a)}_{\vert k}\otimes \pi^{(b)}_k$ where $\left\{p_k\right\}$ is a probability distribution. A state of this form is sometimes referred to as a {\it quantum-classical state}, and the $\rho^{(a)}_{\vert k}$ are sometimes denoted conditional density matrices \cite{Modi12}.

Consider the unnormalised conditional density matrix,
\begin{equation}
\tilde{\rho}^{(a)}_{k}=p_k \rho^{(a)}_{\vert k}.
\end{equation}
This operator satisfies an equation reminiscent of an eigenvalue equation,
\begin{equation}
\rho \pi^{(b)}_k = \tilde{\rho}^{(a)}_{k} \otimes \pi^{(b)}_k,
\end{equation}
with $\tilde{\rho}^{(a)}_{k}$ playing the role of the eigenvalue and $\pi^{(b)}_k$ the role of the eigenprojector. We shall see that in fact these operators do satisfy many of the properties of eigenvalues and eigenvectors respectively, and thereby provide a generalisation of them.

\begin{definition} A {\it subsystem operator-valued eigenvalue} (SOVE) $\tilde{\rho}^{(a)}$ and corresponding {\it subsystem eigenprojector} (SE) $\pi^{(b)}$  is any pair of such operators that satisfy,
\begin{equation}
\rho \left(\openone^{(a)} \otimes\pi^{(b)}_k\right)
= \left(\openone^{(a)} \otimes \pi^{(b)}_k\right) \rho
= \tilde{\rho}_k^{(a)} \otimes \pi^{(b)}_k,
\end{equation}
where $\pi^{(b)}$ is assumed have rank one, $\tilde{\rho}^{(a)}$ has support solely on $a$ and $\pi^{(b)}$ has support solely on $b$.
\end{definition}
For a state $\rho$ which is classical with respect to sub-system $b$, a SOVE can be computed from any corresponding SE via the equation
$\tilde{\rho}^{(a)}_k=\textrm{Tr}_{b}\,\left[(\openone\otimes \pi^{(b)}_{k}) \rho \right]$.

\begin{lemma}\label{SOVElemma}
Suppose that $\tilde{\rho}^{(a)}_1$ and $\tilde{\rho}^{(a)}_2$ are SOVEs with corresponding SEs $\pi^{(b)}_1$ and $\pi^{(b)}_2$.  If the SOVEs are distinct then the SEs must be orthogonal.
\end{lemma}
\begin{proof}
The proof is identical to a well-known proof of the orthogonality of eigenprojectors with distinct eigenvalues, and we include it for completeness:
We have $\rho \left(\openone^{(a)} \otimes \pi^{(b)}_1\right) = \tilde{\rho}^{(a)}_1 \otimes \pi^{(b)}_1$ and $ \left(\openone^{(a)} \otimes \pi^{(b)}_2\right) \rho = \tilde{\rho}^{(a)}_2 \otimes \pi^{(b)}_2$. Hence, $\tilde{\rho}^{(a)}_1 \otimes \textrm{Tr}\left(\pi^{(b)}_2\pi^{(b)}_1\right)=\tilde{\rho}^{(a)}_2 \otimes \textrm{Tr}\left(\pi^{(b)}_2\pi^{(b)}_1\right)$ but if $\tilde{\rho}^{(a)}_1\neq \tilde{\rho}_2$, then the only solution to this equation is $\textrm{Tr} \left (\pi^{(b)}_2\pi^{(b)}_1 \right)=0$.
\end{proof}

For any given quantum-classical state $\rho=\sum_k \pi^{(b)}_k \rho \pi^{(b)}_k$, the SOVEs can be \textit{degenerate}, i.e. there can be two (or more) rank-one projectors $\pi^{(b)}_k$ and $\pi^{(b)}_{k'}$ such that  $p_k \rho^{(a)}_{\vert k}=p_{k'} \rho^{(a)}_{\vert k'}$.  Note however that the decomposition of $\rho$ here has a form reminiscent of a spectral decomposition of a Hermitian operator into eigenvalues and eigenprojectors.  An elementary result is that sets of orthogonal rank-one eigenprojectors of Hermitian operators are not unique when the spectrum includes degenerate eigenvalues, and that uniqueness is recovered when rank-one eigenprojectors are combined into full-rank eigenprojectors, corresponding to maximal subsets of rank-one eigenprojectors for distinct eigenvalues.
\begin{remark}
For any finite-dimensional Hermitian operator $\rho$, there is a \textit{unique} set of full-rank projectors $\Pi_{k}$ such that,
$\rho =\sum_{k}\textrm{Tr}\left( \rho \Pi_{k}\right) \Pi_{k}$, which also satisfy $\sum_{k}\Pi_{k}=\openone$ and
$\Pi_{k}\rho =\rho \Pi_{k}=\textrm{Tr}\left( \rho \Pi_{k}\right) \Pi_{k}$.
\end{remark}
Here the uniqueness follows from the full-rank property, and the orthogonality of the eigenprojectors associated with different eigenvalues. Alternatively, it follows directly as a corollory of Lemma \ref{lemma:UniqueQCDecomposition} below.


Proceeding now by analogy, we make the following definitions:
\begin{definition}
Let $\rho$ be quantum-classical state with respect to a bipartition into subsystems $a$ and $b$.  We define a {\it Full-Rank Subsystem Eigenprojector (FRASE)} for $\rho$ to be any SE $\Pi^{(b)}_{k}$, with rank $\geq1$, for which there does not exist any $ \pi^{(b)}$ on $b$ such that
$\rho \,\left [\openone\otimes\left( \Pi^{(b)}_{k}\!+\!\pi^{(b)}\right)\right]
=\tilde{\rho}^{(a)}_{k}\otimes \left(\Pi^{(b)}_{k}\!+\!\pi^{(b)}\right)$,
where
$\tilde{\rho}^{(a)}_{k}=\textrm{Tr}_b\,\left[\left(\openone\otimes \Pi^{(b)}_{k}\right) \rho \right]/\textrm{Tr}_{b}\left(\Pi^{(b)}_{k}\right)$.
Then we call a decomposition of the form,
\begin{equation}
\label{eq:SubsystemEigenprojectorDeomposition}
\rho =\sum_{k}\tilde{\rho}^{(a)}_{k}\otimes \Pi^{(b)}_{k}
\end{equation}
a {\it FRASE decomposition}, where $\Pi^{(b)}_{k}$ are orthogonal FRASEs satisfying $\sum_{k}\Pi^{(b)}_{k}=\openone$, and every $\tilde{\rho}^{(a)}_k$ is a distinct Hermitian operator on $a$.
\end{definition}
It is now straightforward to prove that FRASE decompositions can be made along similar lines to expansions of Hermitian operators in their eigenvalues and full-rank eigenprojectors:
\begin{lemma}
\label{lemma:UniqueQCDecomposition}
Every $\rho$ which is quantum-classical with respect to a bipartition into subsystems $a$ and $b$, possesses a unique FRASE decomposition as given by Eq.~(\ref{eq:SubsystemEigenprojectorDeomposition}).
\end{lemma}
\begin{proof}
The existence of a FRASE decomposition for any quantum-classical state follows directly from the definition for these states given above (by combining SEs for degenereate SOVEs).  The uniqueness follows immediately from the fact that each FRASE is a full-rank projector onto a subspace, and a full-rank projector onto a sub-space is unique.

\end{proof}
FRASEs satisfy many similar properties to full-rank eigenprojectors, and the standard definition of eigenprojectors is recovered as $b$ is extended to the whole system. The uniqueness of FRASE decomposition underpins the simulation methods in this paper.

Now Lemma~\ref{lemma:UniqueQCDecomposition} above gives rise to the following corollary for concordant states, which provides a useful uniqueness argument which will be needed later on:
\begin{corollary}
\label{corollary:DecompositionConcordantStates}
Any subset $b$ of qudits in a concordant state $\rho=\sum_{\bf x} L \,p\!\left({\bf x}\right) \!\vert{\bf x}\rangle\!\langle {\bf x} \vert L^\dag$, where ${\bf x}=\{x^{(1)},x^{(2)},\cdots\}$ labels the computational basis and $L=L^{(1)}\!\otimes L^{(2)\cdots}$ denotes local-unitary rotations for every qudit, has a unique set of FRASEs $\{\Pi^{(b)}_{k}\}$ which possess a product basis (that is to say the FRASEs are a sum of orthogonal product states).
\end{corollary}
\begin{proof}
The restrictions of the components of $\rho$ to the subset of qudits in $b$,
$L^{(b)}\!\left\vert {\bf x^{(b)}}\right\rangle\!\!\left \langle {\bf x^{(b)}} \right\vert \! L^{(b)\dag}$,
provide an orthogonal set of subsystem eigenprojectors from which a unique set of FRASEs can be constructed following Lemma~\ref{lemma:UniqueQCDecomposition}.
\end{proof}
It is important to note that the local basis for $b$ itself may not be unique, as is the case for example when the subsystem-eigenprojector decomposition yields one FRASE which is the maximally-mixed state.

\subsection{Monte-Carlo simulation of concordant computation}
\label{Sec:MonteCarloForConcordantComputation}

Now we formalize the central challenge tackled in this paper and Ref.~\cite{Eastin10}.  Our notation is as follows: We are given unitary circuit $C$, consisting of unitary gates $G_{t}$ for the $t^{\it th}$ time step, on a system of $N$ qudits (with arbitrary dimensions).  The partially-completed unitary after the $t^{\rm th}$ step is $U_{t}=\mathcal{T}\Pi _{k=1}^{t}G_{k}$ (where $\mathcal{T}$ denotes that the product respects the temporal ordering of the unitaries).  We are also given $\rho_0$ which must admit a polynomially-sized description, and the promise that the state of the system at every step, $\rho_t=U_t \rho_0 U^\dag_t$ (where $\rho_{0}$ is of the form
$\rho_t= \sum_{\bf x} p_t({\bf x}) L_{t}\vert {\bf x} \rangle\!\langle {\bf x} \vert L_{t}^\dag$
 and where $L_0$ and $p_0$ only are given as part of the specification of the problem).  Then the overarching goal can be stated as: Find an equivalent circuit $C^\prime$, consisting of unitary gates $G_{t}^\prime$ with partial completion of the circuit $U_{t}^{\prime }=\mathcal{T}\Pi _{k=1}^{t}G_{k}^\prime$, such that $U_{t}\rho _{0}U_{t}^{\dag }=U_{t}^{\prime }\rho _{0}U_{t}^{\prime \dag }$, but for which there are corresponding pure-state trajectories which are known to be efficiently simulable.

Using the results of Sec.~\ref{Sec:SubsystemEigenprojectorDecomposition}, we can now proceed to derive the general form of a circuit suitable for simulating $C$.  At time step $t$, $G_t$ defines a bipartition of the system qudits into its support $b$ and the rest $a$.  The system state after $t-1$ has form
\begin{equation}
\label{eq:degeneracydiagnosed}
\rho_{t-1} = \sum_k \tilde{\rho}^{(a)}_k\otimes \Pi^{(b)}_k
\end{equation}
where
$\tilde{\rho}_k^{(a)}= \textrm{Tr}_{b} \left[\left(\openone^{(a)} \otimes \Pi^{(b)}_{k}\right) \rho_{t-1}\right] \,/\, \textrm{Tr}_{b}\left(\Pi^{(b)}_{k}\right)$
and
$\left\{\Pi^{(b)}_k\right\}$ is the unique set of FRASEs following Corollary~\ref{corollary:DecompositionConcordantStates}, which are sums of orthogonal product states $L^{(b)}_{t-1}\left\vert {\bf x^{(b)}}\right\rangle\!\!\left \langle {\bf x^{(b)}}\right\vert L^{(b)\dag}_{t-1}$.   Then, at the end of time step $t$
\begin{eqnarray}
\rho _{t}&=& G_t \rho_{t-1} G_t^\dag \nonumber \\
         &=&
\sum_{k}\tilde{\rho}^{(a)}_{k}\otimes G_{t} \Pi^{(b)}_{k} G_{t}^{\dag }.
\end{eqnarray}
By inspection, the operators $G_{t} \Pi^{(b)}_{k} G_{t}^{\dag }$ are subsystem eigenprojectors of $\rho_t$ for the same bipartition. Since the $\Pi^{(b)}_k$ are FRASEs for $\rho_{t-1}$, and $G_t$ is unitary, it follows that the $G_t \Pi^{(b)}_k G^{t\dag}$ are FRASEs for $\rho_t$.  Furthermore since $\rho_t$ is concordant, the operators $G_t \Pi^{(b)}_k G^{t\dag}$ must be a sum of orthogonal product states $L^{(b)}_{t}\left\vert {\bf x^{(b)}}\right\rangle\!\!\left \langle {\bf x^{(b)}}\right\vert L^{(b)\dag}_{t}$ by Corollary~\ref{corollary:DecompositionConcordantStates}.  The uniqueness property of FRASEs now gives us the following key equation:
\begin{equation}
\label{eq:GPGisLDLPLDL}
G_t \Pi^{(b)}_k G_t^{\dag }
=
L_{t}^{(b)} D_{t} L_{t-1}^{(b)\dag} \Pi^{(b)}_k
L_{t-1}^{(b)} D_{t}^\dag L_{t}^{(b)\dag }
\,\,\,\,\,\forall k,
\end{equation}
where $D_t$ accounts for the possibility of a permutation of the computational-basis states on $b$.
(Note that there is typically freedom in choices for $D_t$.)

For the system state we have the equivalence,
\begin{eqnarray}
\label{eq:singlestepupdate}
\rho_t
&=&G_t \rho_{t-1} G^{\dag}_t \nonumber \\
&=&L^{(b)}_t D_t L^{(b)\dag}_{t-1} \rho_{t-1} L^{(b)}_{t-1} D_t^\dag L^{(b)\dag}_t.
\end{eqnarray}
From this equation it should be noted that $L_{t}$ agrees with $L_{t-1}$ other than (possibly) on the $b$,
and that $p_{t}\left({\bf x}\right) =p_{t-1}\left(\{ {\bf y},D_t^{-1}({\bf z})\}\right)$
(where ${\bf z}$ is the part of ${\bf x}$ in $b$, ${\bf y}$ is the part of ${\bf x}$ not in $b$ and we write $D_t\vert{\bf{z}}\rangle=\vert{D_t(\bf{z})}\rangle$).
For $C$ as a whole we have,
\begin{eqnarray}
\label{eq:equivalentcircuit}
\rho_t
&=& U_t \rho_0 U_t^\dag \nonumber \\
&=&
\{\mathcal{T}\Pi_{k=1}^t (L_k D_k L^{\dag}_{k-1})\}
\rho_0
\{\mathcal{T}\Pi_{k=1}^t (L_k D_k L^{\dag}_{k-1})\}
^\dag \nonumber \\
&=&
L_t
\{\mathcal{T}\Pi_{k=1}^t  D_k \}
L^{\dag}_0
\rho_0
L_0
\{\mathcal{T}\Pi_{k=1}^t  D_k \}^\dag
L^{\dag}_{t}.
\end{eqnarray}

To find a complete simulation algorithm for concordant computation, the challenge now is to find the gates $L_k$ and $D_k$ (for all time steps) which make up $C^\prime$, given the initial state and $C$. We will return to this challenge shortly however, and consider how the output statistics would be simulated supposing, for argument's sake, that $C^\prime$ has already been found.  First, we observe that a simulation algorithm does not need to compute an explicit description of the full state at every time step, but can instead just record changes to the system state using an update rule in keeping with Eq.~(\ref{eq:singlestepupdate}). A suitable update rule is: (i) record a new local basis, for every qudit in the support of the gate which acts, specified by a complete set of rank-1 projectors (which need not be unique); (ii) record a suitable permutation operator which acts on the support of the gate.  Given such an update a rule, the output statistics of $C'$ would be sampled as follows:
\begin{remark}
Monte-Carlo simulation, using a pre-calculated update rule, of the output statistics for a given concordant computation with initial state
$\rho_0 = \bigotimes_{j=1}^{N}
L^{(j)}_0 \!
\left(
p^{(j)}_{0}\!(0) \vert 0\rangle\!\langle 0 \vert
\!+\!
p^{(j)}_{0}\! (1) \vert 1\rangle\!\langle 1 \vert
\!+\!
\cdots
\right)
\!
L^{(j)\dag}_0
$:
\begin{itemize}
\item{Define start of a stochastic trajectory, (in the computational basis), by randomly picking a $N$-digit qudit string $s_{\rm in}$ according to probability distributions $p^{(1)}_0(\cdot),\cdots,p^{(N)}_0(\cdot)$.}

\item Permute $s_{\rm in}\mapsto s_{\rm out}$ using $\mathcal{T}\Pi_{k=1}^{t_f} D_k$ given by the update rule.

\item Sample probabilities
$\textrm{Tr} \!\left( \vert b^{(m_j)}\rangle\!\langle b^{(m_j)}\vert L^{(m_j)}_t \vert s_{\rm out}^{(m_j)}\rangle\!\langle s_{\rm out}^{(m_j)}\vert L^{(m_j)\dag}\right)$
on specified qudits $m_j$ with measurement bases $\{\vert b^{(m_j)} \rangle\}$. This is equivalent to stochastically flipping some of the bits of $s_{\rm out}$ with probabilities defined by the measurement basis.

\item Repeat, and gather statistics.
\end{itemize}
\end{remark}
Note that the local-basis changes at intermediate steps are not required to define the trajectories---only the basis of the final measurements.  However, the
intermediate local-basis projectors do play an essential role elsewhere for deriving the update rule from the initial specification of the concordant computation (as tackled in Sec.~\ref{Sec:GateStructureAndLBF}).

The simulation algorithm described in Ref.~\cite{Eastin10} outlines a procedure for finding $C^\prime$ from the specification of a concordant computation. For each time step $t$, the simulation algorithm works in three stages: The first stage is equivalent to finding the decomposition for $\rho_{t-1}$ in Eq.~(\ref{eq:degeneracydiagnosed}) (termed ``diagnosing the degeneracy'' in the reference), and is done by testing for the (permutation) symmetries of $L_0^\dag \rho_0 L_0$ on the support of the partially-completed circuit $U_t$.  The second and third stages were not described in detail in the reference, but are loosely equivalent to solving our Eq.~(\ref{eq:GPGisLDLPLDL}) for $L_t^{(b)}$ and then $D_t$.  The first stage however constitutes an NP-hard problem in general (as explained in Ref.~\cite{Eastin10}), which undermines the success of the simulation algorithm.

One way around this is to place a restriction on $C$, so that $\rho_0$ only needs to be tested for a limited group of symmetries.  Ref.~\cite{Eastin10} made the stipulation that each $G_k$ should have support on one or two qubits only, as in this case the test can be limited to the set of classical-reversible gates which are {\it linear} \cite{NielsenChuang,PreskillLectureNotes}.
Ref.~\cite{Eastin10} includes an efficient symmetry test for this case.  Eastin's proof of efficiency is not written using standard quantum information techniques. To aid the reader, therefore, we present an alternative formulation of this result here. We  present a theorem and corollary, of which the latter is equivalent to Lemma 4 in \cite{Eastin10}.

As we show below, the efficiency of Eastin's test can be attributed to the fact that one- and two-qubit reversible classical gates (CNOT gates, NOT gates, and combinations) are in the Clifford group. After stating the more general Theorem~\ref{theorem:EfficientCliffordSymmetryTesting} we then derive Corollary~\ref{corollary:EfficientSimulationForOneTwoQubitGates},  equivalent to  Lemma 4 in \cite{Eastin10}.

\begin{theorem}
\label{theorem:EfficientCliffordSymmetryTesting}
Suppose that $L_0^\dag\rho_0 L_0=\bigotimes_{j=1}^N\left( L_{0}^{\left( j\right) \dag }\rho _0^{\left( j\right) }L_0^{\left(j\right)}\right)$ as above, which is factorised and diagonal in the computational basis, and that $S_\sigma$ is Clifford unitary on the $N$ qubits. Then: $S_\sigma L_0^\dag \rho _0 L_0 S_\sigma^\dag=L_0^\dag\rho_0 L_0$ if and only if $\textrm{Tr}\left[ \left( S_\sigma L_0^\dag\rho _0 L_0 S_\sigma^\dag - L_0^\dag\rho_0 L_0\right) Z^{\left( j\right) }\otimes \openone^{\left( \backslash j\right) }\right] =0\,\,\,\forall j$. The expectation values $\textrm{Tr}\left[ \left( S_\sigma L_{0}^\dag \rho _0 L_0 S_\sigma^\dag\right) Z^{\left( j\right)}\otimes \openone^{\left( \backslash j\right) }\right] $ and $\textrm{Tr}\left[ \left( L_0^\dag\rho_0 L_0 \right) Z^{\left( j\right) }\otimes \openone^{\left( \backslash j\right) }\right] $ can be computed efficiently, and the overall computational complexity for evaluating all the required expectation values scales quadratically with $N$.
\end{theorem}
\begin{proof}
See Appendix~\ref{Sec:AppendixForTwo}.
\end{proof}

\begin{corollary}
\label{corollary:EfficientSimulationForOneTwoQubitGates}
FRASES can be computed efficiently for every step in a concordant computation on $N$ qubits, for which the circuit $C$ is composed entirely of one and two-qubit gates, and $L_0^\dag\rho _0 L_0$ is factorised and diagonal in the computational basis.
\end{corollary}

\begin{proof}
At time step $t$, it is necessary to find the FRASES for $\rho_{t-1}$ on the support $b$ of $G_t$, and is equivalent to finding the full symmetry group of $\rho _{t-1}$ on $b$.  The promise of concordant computation implies that $\rho _{t-1}=\sum \tilde{\rho}_{k}^{\left( a\right) }\otimes \Pi _{k}^{\left( b\right)}$ where the $\Pi _k^{\left( b\right)}$ are FRASEs, and it is required to find the projectors $L_{t-1}^\dag \Pi _k^{\left( b\right)} L_{t-1}$ in the computational basis ($L_{t-1}$ is known from the previous time step).  This can be done by finding all classical reversible gates $P$ on $b$ satisfying,
$\left( D_{t-1}\cdots D_1\right) ^\dag P\left( D_{t-1}\cdots D_1\right) \left( L_0^\dag \rho _0 L_0\right) \left( D_1\cdots D_{t-1}\right)^{\dag }P\left( D_1\cdots D_{t-1}\right) = L_0^\dag\rho _0 L_0$.
Under the restriction to one and two-qubit gates, all of the classical reversible gates $D_1,\cdots,D_t$ and $P$ are Clifford gates (since they can be generated using only NOT and CNOT gates).  Theorem~\ref{theorem:EfficientCliffordSymmetryTesting} can therefore be applied here, and it  guarantees that all necessary symmetry tests can be performed efficiently.
\end{proof}

\section{New simulation algorithm for concordant computation with gates of arbitrary size}
\label{Sec:GateStructureAndLBF}

The aim of this section is to develop an algorithm which can be used to simulate concordant computations which involve gates of arbitrary size.  We continue using the notation introduced in Sec.~\ref{Sec:MonteCarloForConcordantComputation} for states, gates and circuits.  To recap from Ref.~\cite{Eastin10}, the algorithm therein attempts to identity a FRASE decomposition at every time step (for the initial state $\rho_0$).  As stated above, deriving the FRASE decomposition by considering the whole history of the computation and the symmetries of the input state is NP-hard, equivalent to testing satisfiability for an arbitrary Boolean function.

	The NP-hardness is avoided in Ref.~\cite{Eastin10} by restricting the simulation to concordant circuits composed of gates with support on only one or two qubits. Reversible one- and two-qubit gates are all in the Clifford group, and the efficiency of this algorithm is stated as Corollory \ref{corollary:EfficientSimulationForOneTwoQubitGates}. Here we develop an alternative approach.
	
	The requirement of concordance places strong conditions on every quantum gate. In particular, the FRASE decomposition for the state after a gate $G$ can often be derived from the properties of $G$ alone. Here we develop an  algorithm which exploits this. The algorithm produces output equivalent to the one in Ref.~\cite{Eastin10}, namely a sequence of permutation gates $D_t$ and local-projector changes $L^{(b)}_{t-1} \vert {\bf x}^{(b)} \rangle \! \langle {\bf x}^{(b)} \vert L^{(b)\dag}_{t-1} \mapsto L^{(b)}_t \vert {\bf x}^{(b)} \rangle \! \langle {\bf x}^{(b)} \vert L^{(b)\dag}_{t}$, which are used to construct Monte-Carlo trajectories as described in Sec.~\ref{Sec:MonteCarloForConcordantComputation}.

In Sec.~\ref{Sec:ExplanationOfLBF}, we first derive the general structure of any gate which satisfies the promise of having concordant states for the input and output, and then we will present a general method for solving for its projectors $\{ L^{(b)}_t\vert {\bf x}^{(b)} \rangle\! \langle {\bf x}^{(b)} \vert L^{(b)\dag}_t \}$ and $D_t$ from the given unitary $G_t$ (and previously derived $\{ L_{t-1}\vert {\bf x}^{(b)}\rangle \! \langle {\bf x}^{(b)} \vert L_{t-1}^\dag \}$).  In Sec.~\ref{Sec:ExactnessAndEfficiency}, we explain how each of the steps in the LBF can be implemented using exact arithmetic, and review the computational resources required.

\subsection{Local-basis-update equations and method of solution}
\label{Sec:ExplanationOfLBF}

We start with a technical remark necessary to characterize all gates occurring in concordant computations:
\begin{remark}
Suppose that $\mathcal{B}$ is a partitioning of the labels of the computational-basis states for qudits in $b$.  A unitary $B$ is block diagonal with respect to the partitioning $\mathcal{B}$ of the computational basis of $b$, which is to say that the matrix representation is block-diagonal up to (identical) reordering of the rows and columns, if and only if $B$ commutes with all projectors
$X_j =\sum_{{\bf x}\in \beta_j}\vert {\bf x} \rangle\!\langle {\bf x} \vert$
with
$\beta_j \in \mathcal{B}$.
\end{remark}
\begin{proof}
\begin{equation*}
\langle {\bf x} \vert \left [ B , X_j \right ] \vert {\bf x^\prime} \rangle
=
\begin{cases}
0 \,\,\mbox{if}\,\, {\bf x}, {\bf x}^\prime \in \beta_j, & \\
0 \,\,\mbox{if}\,\, {\bf x}, {\bf x}^\prime \notin \beta_j, & \\
\langle {\bf x} \vert B \vert {\bf x}^\prime \rangle \,\,\mbox{if}\,\, {\bf x} \notin \beta_j, {\bf x}^\prime \in \beta_{j},  & \\
-\langle {\bf x} \vert B \vert {\bf x}^\prime \rangle \,\,\mbox{if}\,\, {\bf x}\in \beta_{j}, {\bf x}^\prime \notin \beta_j.
\end{cases}
\end{equation*}
\end{proof}
Following the line of argument in Sec.~\ref{Sec:MonteCarloForConcordantComputation}, every gate specified in a given concordant computation can be decomposed as follows:
\begin{lemma}
\label{lemma:LDBLGateDecomposition}
Each gate $G_{t}$ specified for concordant circuit $C$ at time step $t$ can be decomposed as
\begin{equation}
\label{eq:gatedecomposition}
G_{t}=L_{t}^{(b)} D_t B_t L_{t-1}^{(b)\dag}
\end{equation}
where $L_{t-1}^{(b)}$ is the local unitary on the support $b$ of $G_t$ after all previous time steps, $D_{t}$ is a classical-reversible gate, and $B_{t}$ is block diagonal, being a direct sum of components which act identically on the projectors $X_k=L_{t-1}^{(b)\dag} \Pi^{(b)}_k L_{t-1}^{(b)}$, where $\{\Pi^{(b)}_k\}$ is the set of FRASEs for $\rho_{t-1}$.
\end{lemma}
\begin{proof}
This follows immediately from Eq.~\ref{eq:GPGisLDLPLDL} writing
\begin{equation*}
\left( D_t^\dag L_t^{(b)\dag} G_t L_{t-1}^{(b)} \right)
\left( L_{t-1}^{(b)\dag}\Pi_k^{(b)}L_{t-1}^{(b)} \right)=
\left( L_{t-1}^{(b)\dag} \Pi_k^{(b)}L_{t-1}^{(b)} \right)
\left( D_t^{\dag}L_t^{(b)\dag}G_t L_{t-1}^{(b)} \right)\,\,\,\,\forall\,k
\end{equation*}
using the Remark above.
\end{proof}

The LBF exploits the guarantee of a decomposition of $G_t$ as by Eq.~(\ref{eq:gatedecomposition}), to solve for $\{ L^{(b)}_t\vert {\bf x}^{(b)} \rangle\! \langle {\bf x}^{(b)} \vert L^{(b)\dag}_t \}$ and $D_t$.  To do this, it runs over all possible computational-basis projectors $X_k$ (of arbitrary rank) on the support $b$ of $G_t$, and attempts to recover rank-one (pure) local-basis projectors by solving the following set of non-linear equations (for each qudit $j$ in $b$) to find unknown local basis projectors $\rho^{(j)}$,
\begin{eqnarray}
\label{eq:lbfequations}
& \left [\openone^{(b/j)} \otimes \rho^{(j)},G_t L^{(b)}_{t-1} X_k L^{(b)\dag}_{t-1} G_t^\dag \right ]=0 {\,\,\,\,\,\,\,\,\,\,\,\rm (i)} \nonumber\\
& \!\!\!\!\!\!\!\!\!\!\!\!  \mbox{subject to,} \nonumber \\
& \textrm{Tr}\,( \rho^{(j)} ) = \textrm{Tr}\,(\rho^{(j)\,2}) = \textrm{Tr}\,(\rho^{(j)\,3})=1 {\,\,\,\,\,\,\,\,\,\,\rm (ii)}
\end{eqnarray}

Solutions to these equations will be  sums of local basis projectors  that are mapped to local projectors by $G$. General solutions to (i), for each qudit, can be arbitrary linear combinations of the desired local projectors, and the constraints (ii) are required to solve for solutions that correspond to pure basis states.  Specifically, constraints (ii) impose purity on general Hermitian operators.  For qubits only the conditions on $\textrm{Tr}\,( \rho^{(j)} )$ and $\textrm{Tr}\,(\rho^{(j)\,2})$ are required, whilst the additional condition on $\textrm{Tr}\,(\rho^{(j)\,3})$ is required when the dimension is greater than two\cite{nickjones}.

Our method for solving Eq.~(\ref{eq:lbfequations}) is as follows: First Gaussian elimination is used to find a general Hermitian solution $\rho^{(j)}$ for Eq.~(\ref{eq:lbfequations})(i).  A random instance $\tilde{\rho}^{(j)}$ of $\rho^{(j)}$ typically has support on the same local-basis projectors as $\rho^{(j)}$.  Hence to
derive rank-one solutions of Eq.~(\ref{eq:lbfequations})(i) and (ii), a random choice is made for $\tilde{\rho}^{(j)}$, the eigenvalues of $\tilde{\rho}^{(j)}$ are found from its characteristic polynomial, and the corresponding eigenvector projectors are derived by back substitution into the eigenvector equation. (The process can be repeated to address rare cases where a bad choice is made for $\tilde{\rho}^{(j)}$.)  For each $X_k$, there are three types of solution to  Eq.~(\ref{eq:lbfequations}):
\begin{enumerate}
\item a complete local basis cannot be found for every qudit of $b$ (in which case
$L^{(b)}_{t-1} X_k L^{(b)\dag}_{t-1}$ does not correspond to a valid input FRASE).
\item a complete local basis of unique rank-one projectors is found for each qudit of $b$.
\item a complete local basis of rank-one projectors is found for each qudit of $b$, where at least some of the projectors are not unique.  (This occurs when local-basis projectors are combined in $G_t L^{(b)}_{t-1} X_k L^{(b)\dag}_{t-1} G_t^\dag$ and there is an infinite family of solutions.)
\end{enumerate}
Note that it is only the rank-one projectors, and not the local unitaries $L_t$, that are needed for the simulation algorithm.  In general the LBF finds a complete local basis only for a subset of the $X_k$, and we call this set $\chi$.  The occurrence of non-unique local projectors could potentially cause difficulties when comparing results for different choices of $X_k$.  For example, when $X_k=\openone$ any complete local basis on $b$ is a solution to Eq.~(\ref{eq:lbfequations}).  To address this, we define the {\it $X_k$-unique} local basis (for each qubit of $b$) as the unique combinations of rank-one projectors of minimal rank which are common to all possible local-basis solutions of Eq.~(\ref{eq:lbfequations}).  In other words, the projectors in a $X_k$-unique local basis are the smallest local-basis projectors which are uniquely determined by Eq.~(\ref{eq:lbfequations}).  We denote the $X_k$-unique local basis projectors by $\openone^{(b/j)} \otimes \rho_u^{(j)}$, and they are easily found - for example by repeatedly solving for the local-basis.  We denote by $LB_k$ the full set of {\it $X_k$-unique} local-basis projectors whenever it is defined.

To implement a Monte-Carlo simulation of the concordant computation, as described in Sec.~\ref{Sec:MonteCarloForConcordantComputation}, there must be an unambiguous update rule for every time step.  An unambiguous update rule can be obtained for time step $t$, only if the LBF is able to identify a unique set of local-basis projectors compatible with all $LB_k$ for $X_k\in\chi$.  For this to be possible, it is necessary that the elements of each $LB_k$ are common projector solutions for all $X_{k^\prime}\in\chi$ --- that is to say that
$\left[\openone^{(b/j)} \otimes \rho_u^{(j)},G_t L^{(b)}_{t-1} X_{k^\prime} L^{(b)\dag}_{t-1} G_t^\dag \right]=0\,\,\, \forall \,\openone^{(b/j)} \otimes \rho_u^{(j)} \in LB_k, X_{k^\prime}\in\chi$ --- and the LBF must test all these conditions.  When these conditions are met we term the $LB_k$ compatible, and the following lemma can be applied:
\begin{lemma}
\label{lemma:SuccessfulLBFOutput}
When the $LB_k$ are compatible $\forall X_k\in\chi$, a complete set of rank-one local-basis projectors solutions can be constructed,
$\left\{ L_t \vert {\bf x}^{(b)} \rangle\!\langle {\bf x}^{(b)} \vert L_t^\dag \vert \right\}$,
for which
$\left[
L_t \vert {\bf x}^{(b)} \rangle\!\langle {\bf x}^{(b)} \vert L_t^\dag
\, , \,
G_t L^{(b)}_{t-1} X_k L^{(b)\dag}_{t-1} G_t^\dag
\right]=0
\,\,\, \forall \,{\bf x}^{(b)}, X_k\in\chi$.
\end{lemma}
\begin{proof}
The set
$LB_{t}^\prime=
\left\{
\text{non-zero projectors}\,
\openone^{(b/j)}\otimes \left( \rho _{u_1}^{(j)}\rho _{u_2}^{(j)}\cdots \rho _{u_{\vert \chi \vert}}^{(j)} \right) \big{\vert} \forall \openone^{(b/j)}\otimes \rho _{u_k}^{\left( j\right) }\in LB_{k} ,\forall j\right\}$.
$LB_{t}^\prime$
is a complete local-basis projector set, and represents a fine-graining of all the $LB_k$, and
$\left[\openone^{(b/j)}\!\otimes\! \rho _u^{(j)},\,G_t L_{t-1}^{(b)} X_{k^\prime} L_{t-1}^{\left( b\right) \dag }G_{t}^{\dag }\right]=0$
$\forall \openone^{(b/j)}\!\otimes\! \rho _u^{(j)}\in LB_t^\prime$,
$X_{k^\prime}\in \chi$.
To find a complete set of rank-one local-basis projectors solutions, the projectors $\rho _u^{\left( j\right) }$ where $\openone^{(b/j)}\otimes \rho _u^{\left( j\right) }\in LB_t$ can be decomposed into rank-one projectors. The vectors that correspond to these rank-one projectors can be obtained by orthogonalising the set of column-vector entries of $\rho _{u}^{\left( j\right) }$ using a Gram--Schmidt process; finally the projectors can be renormalised to have trace $1$.
\end{proof}
A summary of the key steps of our LBF routine in pseudo code is given in Fig.~\ref{fig:LBFPseudoCode}.  The possibility and implications of gates having multiple inconsistent local-basis solutions is taken up in Sec.~\ref{Sec:IncompatibleLBSolns}.

\begin{figure}[t]
\begin{minipage}{0.9\textwidth}
(i) {\bf Pseudo-code for LBF}
\begin{enumerate}
\item INPUT Gate $G_{t}$ with support on qudits in $b$ and the set of rank-one projectors
${L^{(b)}_{t-1}\vert {\bf x}^{(b)} \rangle\!\langle {\bf x}^{(b)} \vert L^{\dag(b)}_{t-1}}$
\item REPEAT For every projector $X_k$ in the computational basis
\subitem SUBROUTINE: SOLVE for local-basis projectors of $G_t L_{t-1} X_k L_{t-1}^{\dag} G^\dag_t$  (see (ii) below).
\subitem IF complete local-basis solution THEN RECORD $X_k$-unique basis in $LB_k$ and $k$ in $EXISTS\!-\!LIST$
\item REPEAT for every $k,k^\prime$ in $EXISTS\!-\!LIST$
\subitem IF all projectors in $LB_k$ commute with $G_t L_{t-1} X_{k^\prime} L_{t-1}^{\dag} G^\dag_t$ DO NOTHING
\subitem ELSE RECORD ``Incompatible solutions''
\item IF ``Incompatible solutions'' THEN OUTPUT ``Local-basis ambiguity at time step $t$'' and STOP
\newline ELSE
\subitem FIND Complete rank-one local-basis projector set compatible with $LB_k$ for all $k$ in $EXISTS\!-\!LIST$
\subitem RECORD rank-one projectors in $LB_t$, and
$\mathcal{B}_t=\{\beta_j\,\vert\,{\rm combining\,over}\,{\bf x}^{(b)}\in\beta_j\,{\rm defines\,unique\,projectors}\}$
\subitem FIND (any) $D_t$ consistent with
$G_t L_{t-1} X_j L_{t-1}^\dag G_t^\dag = L_t D_t X_j D_t^\dag L_t^\dag$ $\forall j$,
where
$X_j=\sum_{{\bf x}^{(b)}\in\beta_j}\vert {\bf x}^{(b)} \rangle\!\langle {\bf x}^{(b)} \vert$
\item OUTPUT $LB_t$ and $D_t$
\end{enumerate}
\vspace{5pt}
(ii) {\bf Pseudo-code for solving local-basis equations}
\begin{enumerate}
\item REPEAT for each qudit $j$ in $(b)$ and for confirmation of the solution
\subitem SET $\rho^{(j)}\propto c_0 \openone + \sum_l c_l \sigma_l$ for Hermitian basis matrices $\sigma_l$ and real expansion coefficients $c_l$
\subitem SET $n_l=[\sigma_l,G_t L_{t-1} X_k L_{t-1}^{\dag} G^\dag_t]$
\subitem SET linear-equation system $\Xi$ as matrix equation $\sum_l c_l n_l = 0$
\subitem APPLY gaussian elimination on $\Xi$ for general Hermitian solution for $\rho^{(j)}$
\subitem PICK random instance $\tilde{\rho}^{(j)}$ of $\rho^{(j)}$ by random choices for free $c_l$
\subitem APPLY root solver on characteristic equation of $\tilde{\rho}^{(j)}$
\subitem APPLY back substitution into eigenvalue equation for rank-one eigenprojectors
$L_t\vert {\bf x}^{(j)}\rangle\!\langle {\bf x}^{(j)} \vert L^\dag_t$ of $\tilde{\rho}^{(b)}$
\item IF complete local basis solution THEN FIND corresponding $X_k$-unique basis following Lemma~\ref{lemma:SuccessfulLBFOutput}
\end{enumerate}
\end{minipage}
\begin{caption}
{\label{fig:LBFPseudoCode} (i) Pseudo-code for LBF: uses promise of existence of a decomposition of gate $G_t$ (following Lemma~\ref{lemma:LDBLGateDecomposition}), to find local projectors $\{L_t\vert {\bf x}^{(b)}\rangle\!\langle {\bf x}^{(b)} \vert\}$ and classical-reversible gate $D_t$ on the support $b$ of $G_t$.  The pseudo-code heralds cases where $G_t$ is consistent with multiple {\it incompatible} (complete) sets of local projectors.  (ii) Pseudo-code for subroutine for solving nonlinear Eqs.~(\ref{eq:lbfequations}) for the local projectors.  The subroutine can be repeated to exclude the possibility of pathological choices for $\tilde{\rho}^{(j)}$.  When a complete local basis is found for all qudits in $b$, the subroutine returns the $X_k$-unique version of it; one way to find the $X_k$-unique basis is by using the promise of projectors with integer matrices (see Sec.~\ref{Sec:ExactnessAndEfficiency}).}
\end{caption}
\end{figure}

\subsection{Implementing the LBF using exact arithmetic}
\label{Sec:ExactnessAndEfficiency}

Errors in the local projectors from time step $t-1$ in our simulation algorithm,
${L_{t-1}^{(b)}\vert{\bf x}^{(b)}\rangle\!\langle{\bf x}^{(b)}|L_{t-1}^{(b)\dag}}$,
can cause a failure to find a complete local basis for $G_{t} L_{t-1}^{(b)} X_{k} L_{t-1}^{(b)\dag}G_t^\dag$
for time step $t$, even though one must exist for the error-free case by the promise of concordant computation.  The simulation algorithm proposed in Sec.~\ref{Sec:GateStructureAndLBF}A has no way of detecting and correcting errors, and they must be prevented from occurring.  To address this issue, we look in detail at implementation of our simulation algorithm using integer arithmetic.  Important goals here are to avoid unreasonable restrictions on the form of the concordant computations which can be simulated using our algorithm, and to ensure that the LBF does not incur excessive demands on computational resources, which should scale polynomially with the number of time steps with reasonable constraints on number size and memory.  We permit irrational numbers in our simulation algorithm when they can be handled using (integer-based) exact arithmetic, and we have found it necessary to involve computations on (irrational) algebraic numbers for some intermediate steps.

First we modify the definitions of concordant states and concordant computation used thus far.  We call a gate or projector rational if it its matrix representation in the computational basis has only rational entries.  Augmenting the definition of a concordant state given in Sec.~\ref{Sec:SubsystemEigenprojectorDecomposition}, we define a concordant state $\rho$ to be {\it rationally concordant} if every subsystem possesses a complete set of rational orthogonal rank-1 projectors $\pi^{(j)}_{k_j}$ such that $\rho=\sum_{k_1,k_2,\cdots} p\!\left({k_1,k_2,\cdots}\right) \pi^{(1)}_{k_1}\otimes\pi^{(2)}_{k_2}\cdots$, and $p\!\left({k_1,k_2,\cdots}\right)$ is a rational probability distribution.  Then we can define a {\it rationally-concordant} computation as a concordant computation for which the system states are also rationally concordant at every time step, and in addition the projectors and probability distributions defining the initial state are rational, as are the gates $G_t$ for all time steps.  As an aside, we point out that our choice to use projectors to represent local bases in our simulation algorithm, (rather than the matrices $L^{(j)}_t$ themselves), prevents many standard gates and states from being excluded by the definition of {\it rationally-concordant} computation here.  Our aim is to involve local rotations which are proportional to (complex)-integer matrices but (generically) have irrational (surd) normalisation factors, such as the Hadamard gate.  For projectors defined using (complex) integer or rational entries, normalization proceeds by dividing out the trace, and surds are avoided.  However a gate such as the $\pi/8$ gate, which is
$\left(\begin{smallmatrix} 1 & 0 \\ 0 & (1+I)/\sqrt{2} \end{smallmatrix} \right)$, must be excluded whenever it would generate surd factors between entries of a local projector occurring in the simulation.

\begin{lemma}
\label{lemma:LBFExactImplementation}
An implementation of the LBF described in Sec.~\ref{Sec:ExplanationOfLBF} using exact arithmetic finds a complete set of rational rank-one projectors for
$\left\{ L_{t}^{(b)}\left\vert {\bf x}^{(b)}\right\rangle \!\! \left\langle {\bf x}^{(b)}\right\vert L_{t}^{(b)\dag }\right\}$,
given rational $G_t$ and rational rank-1 local basis projectors
$\left\{ L_{t-1}^{(b)}\left\vert {\bf x}^{(b)}\right\rangle \!\! \left\langle {\bf x}^{(b)}\right\vert L_{t-1}^{(b)\dag }\right\}$,
provided all possible local basis solutions for the gate are compatible.  Computations using algebraic numbers can be required at intermediate steps.
\end{lemma}
\begin{proof}
Part (i) is for the LBF applied to a single projector, $G_{t}L_{t-1}X_{k}L_{t-1}^{\dag }G_{t}^{\dag}$.  Part (ii) is for finding $LB_t$ from the $LB_k$ when they are compatible, and for resolving higher-rank projectors into rational rank-1 projectors.

(i) We refer to Fig.~\ref{fig:LBFPseudoCode}(ii) for the steps involved in solving for $X_k$-unique local basis projectors, and we give an implementation for them using integer computations: By making integer choices for the Hermitian basis matrices $\sigma_l$, an integer system of equations $\Xi$ can be obtained (for a specific qudit $j$).  A Gaussian elimination method can be applied to $\Xi$ to solve for the general integer Hermitian solution $\rho^{(j)}$.  More specifically, the Hermite normal form for $\Xi$, (an analogue of reduced echelon form for matrices over the integers), can be obtained in polynomial time using Bareiss's algorithm, without suffering an exponential blowup in the memory requirements (see chapter 10 of Ref.~\cite{CheeYap}).  A random choice for (Hermitian) $\tilde{\rho}^{(j)}$ can be made which is an integer matrix (by integer choices of the free variables post Gaussian elimination).

The characteristic equation for the eigenvalues of $\tilde{\rho}^{(j)}$ is then a (real) monic polynomial with integer coefficients, and its solutions must be real.  The elementary rational root test for polynomials dictates that the roots are either integer or irrational algebraic numbers.  Both integer and irrational roots play an essential role for finding local-basis projectors.  Hence we note that exact arithmetic operations can be performed on algebraic numbers using only integer/rational computations.  This can be done by manipulations of polynomials defining the algebraic numbers, for example using an encoding for which the polynomials are represented by companion matrices and field operations are performed using matrix manipulations (see Ref.~\cite{LiLutzer04} for an introductory discussion on this).

A method based on Sturm's theorem can be used to find the eigenvalues of $\tilde{\rho}^{(j)}$ (for a treatment of Sturm's theorem see chapter 7 of \cite{CheeYap}). This theorem can be applied to the characteristic equation for $\tilde{\rho}^{(j)}$ to find the number of distinct roots in any arbitrary interval $(I_1,I_2]$, by using a Sturm sequence for the characteristic polynomial.  More specifically, the number of roots in the interval is given by the difference in the number of sign changes for the values of Sturm sequence when evaluated at $I_1$ and $I_2$. The eigenvalues of $\tilde{\rho}^{(j)}$ can be found by a simple search method which repeatedly bisects a starting interval, at each step selecting one half interval which contains at least one root. This search method finds the eigenvalues exactly when they are integer, and it generates an isolating interval when the eigenvalues are irrational.

Once the eigenvalues of $\tilde{\rho}^{(j)}$ have been found, back substitution is used to find the rank-1 projector solutions. These solutions must be renormalised to have trace value 1. Integer eigenvalues lead to rational eigenprojectors, which are already part of the required $X_{k}-$unique local-basis solution. The existence of rational local-basis projector solutions with rank greater leads to eigenvalues which are irrational, and the associated rank-1 eigenprojectors must also contain irrational numbers.  It is necessary to test combinations of these rank-1 eigenprojectors to find higher-rank projectors which are rational overall. The minimal-rank rational projectors formed this way must be added to the $X_{k}$-unique local-basis solution.  The promise of rational concordant computation guarantees that a complete rational local basis can be found for at least one $X_{k}$.

(ii)When the $LB_k$ are compatible, we can find a fined-grained complete local-basis projector set with elements of rank$\geq 1$ following the Proof of Lemma~\ref{lemma:SuccessfulLBFOutput} above. It is necessary to verify that higher-rank rational projectors can be decomposed into rank-one projectors which are also rational. For this we employ a modified form of Gram-Schmidt as follows: Let $v_1$,$v_2$,$\cdots$ be the column vectors of projector $\rho_u^{\left( j\right)}$. The vectors $v^\prime_1$,$v^\prime_2$,$\cdots$ defined as,
\begin{eqnarray}
v_1^\prime&=&v_1 \nonumber \\
v_2^\prime&=&\langle v_1^\prime ,v_1^\prime \rangle v_2 \!-\!
\langle v_2,v_1^\prime \rangle^\ast v_1^\prime \nonumber \\
v_3^\prime&=&
\langle v_2^\prime ,v_2^\prime\rangle \langle v_1^\prime ,v_1^\prime \rangle v_3 \!-\!
\langle v_3,v_2^\prime\rangle^\ast \langle v_1^\prime,v_1^\prime \rangle v_2^\prime \!-\!
\langle v_3,v_1^\prime\rangle^\ast \langle v_2^\prime ,v_2^\prime \rangle v_1^\prime \nonumber \\
v_4^\prime&=&
\langle v_3^\prime ,v_3^\prime \rangle \langle v_2^\prime ,v_2^\prime \rangle \langle v_1^\prime ,v_1^\prime \rangle v_4\!-\!
\langle v_4,v_3^\prime\rangle^\ast \langle v_2^\prime ,v_2^\prime \rangle \langle v_1^\prime ,v_1^\prime \rangle v_3^\prime\!-\!
\langle v_4,v_2^\prime\rangle^\ast \langle v_3^\prime ,v_3^\prime \rangle \langle v_1^\prime ,v_1^\prime \rangle v_2^\prime\!-\!
\langle v_4,v_1^\prime\rangle^\ast \langle v_3^\prime ,v_3^\prime \rangle \langle v_2^\prime ,v_2^\prime \rangle v_1^\prime \nonumber \\
&& {\rm etc.} \nonumber
\end{eqnarray}
are rational and orthogonal.  The rank-one rational projectors
$
\vert v_1 \rangle\!\langle v_1 \vert/Tr({\vert v_1 \rangle\!\langle v_1 \vert}),
\vert v_2 \rangle\!\langle v_2 \vert/Tr({\vert v_2 \rangle\!\langle v_2 \vert}),
\cdots
$
give the required decomposition into rank-one projectors.
\end{proof}

Next we consider the computational requirements for the LBF implementation described in Lemma.~\ref{lemma:LBFExactImplementation} (for exact computation).  First of all, we observe that the (worst-case) computational overhead for the LBF scales poorly with gate size: For a gate with support having dimension $d$, the total number of projectors of all ranks to which the LBF might be applied scales as $O(2^d)$, a scaling which is doubly exponential with respect to the number of qudits.  In the worst case, the LBF finds a local-basis solution for only one pair of input projectors, and the LBF must be applied to input projectors of all rank between $1$ and $[d/2]$ (note that $\Pi $ and $\openone-\Pi $ must share the same local basis). (In the simplest case, the LBF is applied first to all possible one-dimensional projectors and the output projectors are found to carry the same local-basis, in which case it not necessary to check higher-rank input projectors.)  Hence when considering the computational complexity, we consider the dependence on circuit size for a fixed maximum gate size.

Focusing now on the complexity for computations performed by the LBF given a specific choice of input projector, we can see that all steps involved can be performed efficiently.  The key mathematical steps used by the LBF are: Gaussian elimination (and back substitution), for which complexity scales polynomially with respect to the matrices involved, and root finding, which is efficient due to the use of a bisection method.  Furthermore, the majority of the calculations performed by the LBF use only integer matrices.  Where irrational numbers do occur, however, there are large computational overheads due to the need to perform arithmetic operations using algebraic numbers with no loss of precision.  The promise of rational concordant computation ensures that all irrational contributions must cancel for the output.  Consequently, the computational cost for handling algebraic numbers can be regarded as a fixed overhead that does not undermine the efficiency of the LBF, for increasing numbers of gates.

The computational requirements for the LBF are addressed by the following lemma:
\begin{lemma}
\label{lemma:LBFComplexity} The computational complexity to solve for local-basis updates following Lemma~\ref{lemma:LBFExactImplementation} scales, in regard of both time and space (memory), polynomially with respect to the number of circuit gates (for a fixed maximum gate size), and the total bits required to represent each gate and the initial state of each qudit.
\end{lemma}
\begin{proof}
See Appendix~\ref{Sec:AppendixForThree}.
\end{proof}

\section{Concordant computation with gates which are consistent with incompatible choices for the local basis}
\label{Sec:IncompatibleLBSolns}

The LBF cannot be successfully applied in all cases: It is possible that the output of the LBF for a given unitary is not unique. In this case, the LBF generates incompatible multiple solutions (where the notation of compatibility is laid out in Sec.~\ref{Sec:ExplanationOfLBF}).  This causes the simulation algorithm to fail, since an incorrect local basis may be chosen which would then cause the simulation algorithm to generate an entangling trajectory. A linear number of such events would lead to an exponential number of trajectories which would need to be tested, leading to an inefficient algorithm. Here, we explore some cases where this non-uniqueness arises.  Note that the LBF will commonly output multiple outputs in ways which do not disrupt the algorithm, arising for example from reordering of the local-basis projectors. We are not interested in such cases here, since they are easy to identify and unproblematic, and we focus only on cases where the outputs are truly incompatible.

Where the LBF outputs such incompatible solutions, there is  ambiguity for the corresponding local-basis update, and additional information is required to derive valid simulation trajectories.  We have tested our LBF numerically, by applying it to gates of the form $G=LDBL^\prime$ where $L$, $D$, $B$ and $L^\prime$ were generated randomly in keeping with the general form laid out in Lemma~\ref{lemma:LDBLGateDecomposition}.  In our numerical studies, we found that the LBF did not output incompatible solutions in a large number of cases.  However, we have also found special cases where the LBF outputs incompatible solutions for local-basis projectors depending on how one-dimensional projectors for the input are combined, which we now illustrate using examples.

Our first special case is given by the gate $G_{\rm exc.1}$ which maps the computational basis to the basis of Bell states,  $G_{\rm exc.1}$: $\left\vert j,k\right\rangle \mapsto Z^j\otimes X^k\left( \frac{\left\vert 00\right\rangle +\left\vert 11\right\rangle }{\sqrt{2}}\right)$.
It is convenient to write the action of $G_{\rm exc. 1}$ in the Pauli basis:
\begin{eqnarray}
G_{\rm exc. 1}\left\vert 00\right\rangle\!\!\left\langle 00 \right\vert G_{\rm exc. 1}^\dag
&=&\frac{1}{4}\left( \openone+X\otimes X-Y\otimes Y+Z\otimes Z\right) \nonumber \\
G_{\rm exc. 1}\left\vert 01\right\rangle\!\!\left\langle 01 \right\vert G_{\rm exc. 1}^\dag
&=&\frac{1}{4}\left( \openone+X\otimes X+Y\otimes Y-Z\otimes Z\right) \nonumber \\
G_{\rm exc. 1}\left\vert 10\right\rangle\!\!\left\langle 10 \right\vert G_{\rm exc. 1}^\dag
&=&\frac{1}{4}\left( \openone-X\otimes X+Y\otimes Y+Z\otimes Z\right) \nonumber \\
G_{\rm exc. 1}\left\vert 11\right\rangle\!\!\left\langle 11 \right\vert G_{\rm exc. 1}^\dag
&=&\frac{1}{4}\left( \openone-X\otimes X-Y\otimes Y-Z\otimes Z\right)
\end{eqnarray}
Noting that $\left\vert 00\right\rangle\!\!\left\langle 00\right\vert +\left\vert 11\right\rangle\!\! \left\langle 11\right\vert =\frac{1}{2}\left( \openone+ZZ\right) $, and that by local rotations from the $Z$-basis to the $X$ and $Y-$bases also $\left\vert ++\right\rangle \!\!\left\langle ++\right\vert +\left\vert --\right\rangle \!\!\left\langle --\right\vert =\frac{1}{2}\left( \openone+XX\right) $ and $\left\vert +i+i\right\rangle \!\!\left\langle +i+i\right\vert +\left\vert -i-i\right\rangle \!\!\left\langle -i-i\right\vert =\frac{1}{2}\left( \openone+YY\right) $, we can see that $G_{\text{exc.1 }}$ maps rank-two projectors in the computational basis to rank-two FRASEs carrying either the $X$, $Y$ or $Z$ basis for both qubits. (Note that $G_{\rm exc. 1}$must assign the same local basis both to a projector and the difference of that projector with the identity.)

$G_{\rm exc. 1}$ has a rather exceptional structure that exploits special features of the Bell states.  In contrast, a generic class of gates
which generate FRASE's carrying incompatible local-bases is provided by controlled local unitaries.  A simple example for two qubits would be gate
$G_{\rm exc. 2}=cU$ which implements a rotation $U$ on qubit 2 from the computation basis to any-other qubit basis, controlled by qubit 1.  For the set of input FRASES $\left\{\vert 00\rangle\!\langle 00 \vert,\vert 01\rangle\!\langle 01 \vert,\vert 1\rangle\!\langle 1 \vert\otimes \openone \right\}$, $G_{\rm exc. 2}$ outputs FRASES with the computational basis for both qubits.  For the set of input FRASES $\left\{\vert 10\rangle\!\langle 10 \vert,\vert 11\rangle\!\langle 11 \vert,\vert 0\rangle\!\langle 0 \vert\otimes \openone \right\}$, $G_{\rm exc. 2}$ outputs FRASES with the computational basis for qubit 1, and the rotated basis for qubit 2.  Similar examples can be easily constructed for gates with support on arbitrary numbers of qudits, with arbitrary dimensions.

Our LBF routine heralds the occurrence of incompatible local-basis solutions, but is forced to stop in such cases in the absence of additional information concerning the correct solution to choose.  One possibility is to consider restricted instances of concordant computations which use only gates which do not generate incompatible solutions.  When however it is necessary to consider gates which generate incompatible solutions, it is clear that efficient heuristic tests will suffice to resolve local-basis ambiguities in many cases.  Another approach is to apply the LBF to extended sequences of gates with the aim of finding a unique local-basis update overall.  The efficiency of our simulation procedure however is only preserved when the local-basis ambiguity extends over gate sequences which scale logarithmically with respect to the number of circuit gates.

\section{Conclusions}
\label{Sec:Conclusions}

In this paper we contribute several new results on the problem of constructing efficient classical simulations for concordant computation.  These new results include a method to solve for local-basis updates using exact arithmetic, which we prove is efficient.  However, our results fall short of a proof that all instances of concordant computation admit efficient simulation or, on the contrary, that this is impossible in principle.  The fundamental difficulty for any simulation of concordant computation is the need to test properties for the full quantum state.  These tests generically involve exponentially-large matrices (other than in some special cases) and hence are inefficient.  The simulation procedure of Ref.~\cite{Eastin10} involves symmetry tests on the full quantum state.  It was proved by the author that these tests are computationally equivalent to solving 3-SAT, an NP-Complete problem, proving that the simulation cannot be efficient in general.

In this paper we have progressed beyond the simulation procedure in Ref.~\cite{Eastin10}, by attempting to bypass exhaustive symmetry testing on the quantum state.  Our alternative simulation procedure attempts to derive simulation trajectories directly from the circuit which is supplied in the problem.  This approach can work as the concordant promise heavily constrains the structure of the gates that make up the circuit.  Consequently our analysis has focused on our LBF subroutine, which is for deriving local-basis updates directly from the gates and input local-basis projectors.  Our most important contribution is proof that local-basis updates can be computed efficiently using exact arithmetic.  Furthermore, our investigations have uncovered two classes of special gates for which the LBF outputs multiple incompatible choices for the local-basis updates.  In such cases additional information about the quantum state is required to derive valid simulation trajectories.  We leave the problem of characterizing the full set of special gates as an open challenge.  It is also important to determine if these gates can be used to generate examples of concordant computation that cannot be reduced to those in the class of probabilistic reversible classical computation.

An entirely different approach to overcoming the inefficient symmetry tests of Ref.~\cite{Eastin10} would be to replace NP-hard exact symmetry tests on the quantum state with efficient probabilistic sampling.  Using ideas in Ref.~\cite{VandenNest11}, it is possible to devise efficient symmetry tests based on random sampling, where the probability for error is exponentially suppressed.  However, this approach gives rise to two challenges. The first is to understand the effects of errors within the simulation, which have been circumvented in this paper by using exact methods.  The second is to understand cases when probabilistic methods fundamentally cannot work.  Algorithms are highly structured by nature, and typically they are not well modelled as random processes.  It is possible that there are scenarios involving concordant computation which provably require knowledge of rare hard instances to achieve valid output statistics, where the hard instances foil probabilistic tests on the quantum state.  We leave a full analysis of these issues as an open problem.

\begin{acknowledgments}
H. C. acknowledges funding and support for this work by the National Research Foundation and Ministry of Education, Singapore, as well as the University of Bristol.  We thank Bryan Eastin and Kavan Modi for helpful discussions, and B. E. for detailed feedback on the full manuscript.
\end{acknowledgments}

\appendix

\section{Proof of Theorem~\ref{theorem:EfficientCliffordSymmetryTesting}}
\label{Sec:AppendixForTwo}

\noindent \textbf{Theorem \ref{theorem:EfficientCliffordSymmetryTesting}}
Suppose that $L_0^\dag\rho_0 L_0=\otimes _{j=1}^N\left( L_{0}^{\left( j\right) \dag }\rho _0^{\left( j\right) }L_0^{\left(j\right)}\right)$, which is factorised and diagonal in the computational basis, and that $S_\sigma$ is Clifford unitary on the $N$ qubits. Then: $S_\sigma L_0^\dag \rho _0 L_0 S_\sigma^\dag=L_0^\dag\rho_0 L_0$ if and only if $\textrm{Tr}\left[ \left( S_\sigma L_0^\dag\rho _0 L_0 S_\sigma^\dag - L_0^\dag\rho_0 L_0\right) Z^{\left( j\right) }\otimes \openone^{\left( \backslash j\right) }\right] =0\,\,\,\forall j$. The expectation values $\textrm{Tr}\left[ \left( S_\sigma L_{0}^\dag \rho _0 L_0 S_\sigma^\dag\right) Z^{\left( j\right)}\otimes \openone^{\left( \backslash j\right) }\right] $ and $\textrm{Tr}\left[ \left( L_0^\dag\rho_0 L_0 \right) Z^{\left( j\right) }\otimes \openone^{\left( \backslash j\right) }\right] $ can be computed efficiently, and the overall computational complexity for evaluating all the required expectation values scales quadratically with $N$.

\begin{proof}
The forward direction is trivial. For the reverse, we assume that
$\textrm{Tr}\left[ \left( S_{\sigma }L_{0}^{\dag }\rho _0 L_0 S_\sigma^\dag-L_0^\dag\rho_0 L_0\right) Z^{\left( j\right) }\otimes \openone^{\left( \backslash j\right) }\right] =0\,\,\,\forall j$.  It is necessary to establish that $S_\sigma$ leaves all products of $Z$ and $\openone$ operators unchanged.  (Trivally this holds for the identity.)
Since $L_0^\dag \rho _0 L_0$ is diagonal in the computational basis, for each qubit $j$ which is pure,   $\left\vert 0\right\rangle $ (or $\left\vert 1\right\rangle $), $\textrm{Tr}\left[ \left( L_0^\dag \rho _0 L_0\right) Z^{(j)}\otimes \openone^{\left( \backslash j\right) }\right]$ and $\textrm{Tr}\left[ \left( S_\sigma L_0^\dag \rho_0 L_0 S_\sigma^\dag\right) Z^{\left( j\right) }\otimes \openone^{\left( \backslash j\right) }\right]$ are $1$ (or $-1$); furthermore, $\left\vert 0\right\rangle $ (or $\left\vert 1\right\rangle $) is the only possible qubit state which gives $1$ (or -1). Hence equality of $S_\sigma L_0^\dag\rho_0 L_ 0 S_\sigma^\dag$ and $L_0^\dag\rho_0 L_0$ is established for the pure qubits.
For the remaining qubits $L_0^{\left( j\right) \dag }\rho_0^{\left( j\right) }L_0^{\left( j\right) }=\frac{1}{2}\left( \openone^{(j)}+q_j Z^{(j)}\right) $, and we denote the unique values for the $q_{j}$ by $Q_1$,$Q_2,\cdots $, where $1>Q_{1}>$\ $Q_{2}>Q_{3}\cdots >-1$. We also denote corresponding subsets of qubits with $q_j=Q_k$ by $J(Q_k)$.  The effect of $S_\sigma$ on $L_0^\dag \rho_0 L_0$ in the Pauli basis is to permute the expansion coefficients for $L_0^\dag \rho_0 L_0$ amongst the basis elements.

We now consider basis elements corresponding to the largest expansion coefficient $Q_1$ in $L_0^\dag\rho _0 L_0$ and $S_\sigma L_0^\dag\rho _0 L_0 S_\sigma^\dag$.  For each $j\in J(Q_1)$,
$\textrm{Tr} \left[ \left( L_0^\dag\rho _0 L_0 \right) Z^{\left( j\right) }\otimes \openone^{\left( \backslash j\right) } \right]
=
\textrm{Tr} \left[ \left( S_\sigma L_0^\dag\rho_0 L_0 S_\sigma^\dag \right) Z^{\left( j\right)}\otimes \openone^{\left( \backslash j\right)}\right]$
$\implies $
$S_{\sigma }Z^{\left( j\right) }\otimes \openone^{\left( \backslash j\right) }S_{\sigma }^{\dag }=Z^{\left( k\right) }\otimes \openone^{\left( \backslash k\right) }$ where  $q_{k}\in J(Q_{1})$. An analogous statement also holds when there are multiple $Z$ operators for multiple qubits of $J(Q_{1})$, and so $S_\sigma\left[ \otimes _{j\in J(Q_1)}\frac{1}{2}\left( \openone^{(j)}+Q_{1}Z^{(j)}\right) \right] \otimes \openone^{\left(/J(Q_1)\right)}S_\sigma^\dag=\left[ \otimes _{j\in J(Q_1)}\frac{1}{2}\left( \openone^{(j)}+Q_1 Z^{(j)}\right) \right] \otimes \openone^{\left(/J(Q_1)\right) }$ also.

Proceeding now to $Q_2$, a similar argument can be made as for $Q_1$. The following now holds for each $j\in J(Q_{2})$:
$\textrm{Tr}\left[ \left( L_0^\dag \rho_0 L_0 \right) Z^{\left(j\right)}\otimes \openone^{\left( \backslash j\right) }\right]=
\textrm{Tr} \left[ \left( S_\sigma L_0^\dag \rho _0 L_0 S_\sigma^\dag \right) Z^{\left( j\right) }\otimes \openone^{\left( \backslash j\right) }\right] $
$\implies $ $S_{\sigma }Z^{\left( j\right) }\otimes \openone^{\left( \backslash j\right) }S_\sigma^\dag=Z^{\left( k\right) }\otimes \openone^{\left( \backslash k\right) }$ where $q_k\in J(Q_2)$: a possible ambiguity might be considered when the value of $Q_2$ coincides with a power of $Q_1$, but in this case the corresponding basis element is already accounted for in the previous step.  Hence it can be concluded that $S_\sigma\left[ \otimes _{j\in J(Q_2)}\frac{1}{2}\left( \openone^{(j)}+Q_2 Z^{(j)}\right) \right] \otimes \openone^{\left(/J(Q_2)\right)}S_\sigma^\dag=\left[ \otimes _{j\in J(Q_2)}\frac{1}{2}\left( \openone^{(j)}+Q_2 Z^{(j)}\right) \right] \otimes \openone^{\left(/J(Q_2)\right) }$
The same argument can now be applied for $Q_3$. Possible ambiguities which might be considered when the value of $Q_3$ coincides with a product of $Q_1$'s and $Q_2$'s can be discounted, because the corresponding basis elements are accounted for previously.

By induction it follows that $S_\sigma L_0^\dag\rho_0 L_0 S_\sigma^\dag=L_0^\dag \rho_0 L_0$.  Finally, all expectation values that must be computed for the theorem are of the form for a Pauli product operator acting on $L_0^\dag \rho_0 L_0$, each of these can be computed (via Gottesman-Knill theorem \cite{Gottesman98} techniques) with linear complexity in $N$, there are $2N$ such expectation values to compute and therefore the overall scaling is quadratic.
\end{proof}

\section{Proof of Lemma~\ref{lemma:LBFComplexity}}
\label{Sec:AppendixForThree}

\noindent \textbf{Lemma \ref{lemma:LBFComplexity}} The computational complexity to solve for local-basis updates following Lemma~\ref{lemma:LBFExactImplementation} scales, in regard of both time and space (memory), polynomially with respect to the number of circuit gates (for a fixed maximum gate size), and the total bits required to represent each gate and the initial state of each qudit.

\begin{proof}
We will consider an arbitrary time step within the simulation, for which the LBF is applied to gate $G_t$ and qudit $j$, with dimension $d_j$.  The dimensions of the support of $G_t$ is $d$.  Estimates for the computational requirements of the different types of mathematical steps involved are as follows:
\begin{itemize}
\item {\it Arithmetic on integers:}
The computation requirements for multiplication dominate over those for addition. The (time) complexity for multiplying two $l$-digit numbers scales as $O(l^2)$ and the output has $2l$ digits.  For multiplication of matrices involving $l$-digit integer entries (for real and imaginary parts), the digit length for entries of the output matrix scales linearly with $l$, while the time complexity scales as third order in the matrix dimensions.
\item {\it Gaussian elimination/back substitution:}
Gaussian elimination and back substitution are used by the LBF first to find a solution $\tilde{\rho}^{(j)}$ to Eq.~(\ref{eq:lbfequations}),
and then to find eigenprojectors for $\tilde{\rho}^{(j)}$ (given specific eigenvalues).  The complexity for Gaussian elimination (which has cubic scaling with respect to the number of unknowns) dominates that for back substitution (for which the scaling is quadratic).  The linear system $\Xi$ which must be solved to find $\tilde{\rho}^{(j)}$ is overdetermined, and requires Gaussian elimination on a matrix with dimensions $O(d^2)$ by $O(d_j^2)$.  When Bareiss's algorithm is used, the number of elementary steps is $O(d^2 d_j^3)$ and the maximum number size is $O(d_{j}(\log d_{j}+l))$, where $l$ is the maximum number of digits for the matrix entries at the start (see Lecture 10 of \cite{CheeYap}).  Finding eigenprojectors for $\tilde{\rho}^{(j)}$ requires Gaussian elimination to a matrix with dimensions $d_{j}$ by $d_{j}$, which is integer for the case of an integer eigenvalue, but which has contributions which are algebraic numbers when the eigenvalue is an algebraic number.
\item {\it Root finding for integer polynomials of degree $p$:}
Root finding must be used to find the eigenvalues of $\tilde{\rho}^{(j)}$ (for which roots must be found for its characteristic equation of degree $p=d_{j}$), and it is also used for arithmetic operations on algebraic numbers.  Our method for root finding uses Sturm's theorem (see Lecture 7 of Ref.~\cite{CheeYap}.  Sturm's theorem states that the existence of roots within a given interval can be detected by evaluating a Sturm chain of $p+1$ polynomials at both interval endpoints, and taking the difference of the number of sign changes for the chain.  By testing for the existence of roots within each interval, the search region can be repeatedly subdivided to locate the roots.  If $l$ is the maximum digit-length of the polynomial coefficients, then the initial search region can be taken to be of size $O(2^{l})$ (e.g. using Cauchy's bound for polynomial roots).  For integer solutions the smallest search interval has length 2, and the number of bisections required to locate a root is $O(l)$.  For irrational solutions, an additional running time $O(p^2 l^2)$ is sufficient to identify an isolating interval (see Lecture 6 of Ref.~\cite{CheeYap}).
\item {\it Arithmetic on algebraic numbers:}
Gaussian elimination and back substitution must be performed on matrices mixing integers with irrational algebraic numbers when the eigenvalues of $\tilde{\rho}^{(j)}$ are irrational.  One method for performing arithmetic on algebraic numbers is as follows: Every algebraic number can be represented as an integer polynomial which has the number as a root, together with an isolating interval (e.g. see Lecture 6 of Ref.~\cite{CheeYap}). Arithmetic operations (i.e. addition, multiplication, number comparison, etc) can be done by performing simple computations on companion matrices associated with the polynomials (described for example in Ref.~\cite{LiLutzer04}), together with updates to isolating intervals using the bisection method described in above.  There is a large overhead for executing these arithmetic operations due to the need for Kronecker (tensor) product operations on the companion matrices.
\end{itemize}

\end{proof}

\end{document}